%% file: main.tex
\documentclass[11pt]{article}
\usepackage{defs}
\usepackage{todonotes}

\title{The Greedy Superstring Algorithm Achieves Ratio~2\\ for Strings of~Length~6 Already}
\author{
   Nikolai Chukhin
   \thanks{JetBrains Research.
   Email: \url{buyolitsez1951@gmail.com}}
   \and
   Alexander S. Kulikov
   \thanks{JetBrains Research.
   Email: \url{alexander.s.kulikov@gmail.com}}
   \and
   Ivan Mihajlin
   \thanks{JetBrains Research.
   Email: \url{ivmihajlin@gmail.com}}
   \and
   Alexander Smal
   \thanks{JetBrains Research.
   Email: \url{avsmal@gmail.com}}
}
\date{}

\let\problemname=\textrm
\let\problemabbr=\textrm

\begin{document}
	
\maketitle

\input{sections/abstract}
\input{sections/introduction}
\input{sections/preliminaries}
\input{sections/lower_bound_rho_k}
\input{sections/lower_bound_rho_3}

\input{sections/upper_bound_rho_3}

\section*{Acknowledgments}
We thank Thomas Preu for bringing a typo in~\Cref{sec:rhok-lower} to our attention and for helping us improve the~manuscript.

\section*{AI Disclosure}
The~authors used GPT-5.6 Sol Ultra and GPT-5.6 Sol Pro in~tandem throughout this project to~find and refine constructions, discuss proof ideas, and check arguments.
The~investigation of~$\rho_3$ began with a~model-generated construction giving a~lower bound strictly larger than $\frac53$.
Through repeated discussions and iterations with the models, the authors strengthened this construction and developed the upper-bound argument until the two bounds met at~$\rho_3=\frac95$.
The~same iterative process was used to~search for strong lower bounds for larger $k$, eventually yielding constructions that establish $\rho_k\ge2$ for every $k\ge6$.
The~models also assisted with proofreading and the literature review.
The~authors independently verified all arguments and assume full responsibility for content of the paper.

\bibliographystyle{alpha}
\bibliography{refs}

\end{document}

%% file: sections/abstract.tex
\begin{abstract}
	In~the Shortest Common Superstring (SCS) problem, 
	one is~given a~set of~strings and is~asked to~find
	a~string of~minimum length containing each of~the input strings as~a~substring.
	The greedy superstring conjecture states that the following 
	natural greedy algorithm has approximation ratio~$2$:
		while there is~more than one string, select the pair of~strings with the maximum overlap, merge them, and add the merged string back to~the set.
    The greedy algorithm works in linear time and is probably the simplest possible approximation algorithm for SCS. 
    If the conjecture holds, then the
    greedy algorithm also surpasses the approximation guarantees of the best known approximation algorithms.
    The conjecture is~open for $40$~years already and even the approximation ratio~$\rho_k$ in~the special case in~which
	input strings have length~$k$ has not yet been found: for all $k \ge 3$,
	$2-1/k \le \rho_k \le \min\{(k+1)/2, 3.396\}$.
	
	We~prove that already for strings of~length~$6$, the approximation ratio
	of~the greedy algorithm is~at~least~2: $\rho_k \ge 2$ for all $k \ge 6$.
	We~also show that $\rho_3=9/5$, thus completely characterizing the worst-case behavior of~the greedy algorithm for strings of~length~$3$.
\end{abstract}

%% file: sections/introduction.tex
\section{Greedy Superstring Conjecture}\label{sec:introduction}
In~the Shortest Common Superstring problem (SCS, also known~as Shortest Superstring Problem), one is~given a~set of~strings and is~asked to~find
a~string of~minimum length containing each of~the input strings as~a~substring. It~is a~mathematical model of~the genome assembly problem and has applications in~data compression and scheduling problems~\cite{GP14}. 
For this reason,
approximation algorithms for SCS are actively studied for the last $35$~years
(a~recent paper~\cite{Nikolaev24} provides a~list
of~all known bounds). 
The algorithms and their analysis are getting more and more involved, 
the currently best known approximation ratio is~$2.466$~\cite{EMV23}.
At~the same time, 
there~is an~extremely simple greedy algorithm whose 
approximation ratio is~conjectured~\cite{Storer88,TU88,Turner89,BJLTY91}
to~be equal to~$2$: 
	while there is~more than one string, select the pair of~strings with the maximum overlap, merge them, and add the merged string back to~the set.

This greedy algorithm is~simpler, both to~describe and to~implement, than all known approximation algorithms, many of~which are based on~finding an~optimum weight matching in~a~graph. This also makes~it faster than other approximation algorithms:
assuming that the alphabet size is~constant,
it~is possible to~implement the greedy algorithm so~that it~works in~linear time~\cite{Ukkonen90}. The greedy algorithm is~also known 
to~behave well in~practice~\cite{RBT04,Ma09}.
It~is not difficult to~see that the approximation ratio of~the greedy algorithm 
is~at~least~$2$: for an~instance $\{{\tt a}{\tt b}^{k-1}, {\tt b}^k, {\tt b}^{k-1}{\tt a}\}$, the optimum superstring ${\tt a}{\tt b}^k{\tt a}$ has length~$k+2$,
whereas the greedy algorithm may produce a~superstring ${\tt a}{\tt b}^{k-1}{\tt a}{\tt b}^k$ of~length $2k+1$. The currently best known upper bound is~$3.396$~\cite{EMV23}.

For various \NP{}-hard problems, the behavior of~a~natural greedy heuristic is~well understood and is~known for many decades.
Say, for Vertex Cover, an~algorithm ``while there~is an~uncovered edge, select the edge, add both of~its endpoints to~the cover, and remove all edges incident to them'' (folklore, see~\cite[Algorithm~1.2]{Vazirani01}) has approximation ratio $2$.
For Set Cover, an~algorithm
``while there are uncovered elements, select the set covering the largest number of~uncovered elements, add~it to~the solution, and mark its elements as~covered''
has approximation ratio $\ln n$~\cite{Johnson73}.
For List Scheduling, an~algorithm
``while there are unscheduled jobs, assign the next job to~the least loaded machine''
		has approximation ratio $2-1/\mu$ on~$\mu$ machines~\cite{Graham66}.

For this reason, it~is surprising (and perhaps depressing) that the  approximation ratio 
of~such a~remarkably simple greedy heuristic for SCS is~still unknown. 
Yet, here we~are: the greedy superstring conjecture is~open for 40~years already and even the behavior of~the greedy heuristic in~the special case
of~strings of~constant length is~still a~mystery. 
For each $k\ge2$, let $k$-\problemabbr{SCS} denote the restriction in~which every input string has length $k$, and let $\rho_k$ be~the approximation ratio of~the greedy algorithm on~this restriction.
The~restriction $2$-\problemabbr{SCS} is~in~\P{}, whereas $3$-\problemabbr{SCS} is~already \NP{}-hard~\cite{GMS80}.
The~following summarizes the bounds on~$\rho_k$ known before this work.
The classical example $\{{\tt a}{\tt b}^{k-1}, {\tt b}^k, {\tt b}^{k-1}{\tt a}\}$ shows that $\rho_k \ge \frac{2k+1}{k+2}=2-\frac{3}{k+2}$.
Cazaux and Rivals~\cite{CR18} proved that $\rho_k \le \frac{k+1}{2}$
and constructed a~series of~SCS instances showing that $\rho_k \ge 2-\frac{1}{k}$.
They conjectured that, in~fact, $\rho_k=2-\frac{1}{k}$.
Combining known results gives the following picture:
%
%
\[
	1.66 \le \rho_3 \le 2,\quad
	1.75 \le \rho_4 \le 2.5,\quad
	1.8 \le \rho_5 \le 3,\quad
	2-1/k \le \rho_k \le 3.396,\, \text{for all $k \ge 6$.}
\]

\paragraph{Our Contribution.} 
We~prove that already for strings of~length~$6$, the approximation ratio
of~the greedy algorithm is~at~least~2: $\rho_k \ge 2$ for all $k \ge 6$.
This disproves the conjecture by~Cazaux and Rivals~\cite{CR18}.
Then, we~show that $\rho_3=\frac{9}{5}$, thus improving known lower and upper bounds on~$\rho_3$
and providing the first exact value of~at~least one~$\rho_k$ in~the \NP-hard range $k \ge 3$.


%% file: sections/preliminaries.tex
\section{Preliminaries}\label{sec:preliminaries}

All strings that we~consider are finite.
For strings $u,v$, write $u\sqsubseteq v$ if~$u$ occurs in~$v$ as~a~substring.
For strings $p,q$, their \emph{overlap}, denoted by~$\operatorname{ov}(p,q)$, is~the largest $\ell\ge 0$ such that the length-$\ell$ suffix of~$p$ equals the length-$\ell$ prefix of~$q$.
Their \emph{overlap merge} $p\odot q$ is~obtained by~appending to~$p$ the part of~$q$ following this prefix, and hence
\[|p\odot q|=|p|+|q|-\operatorname{ov}(p,q).\]

For a string $s=s_0s_1\cdots s_{m-1}$ of~length $m\ge1$ and integers $1 \le k \le m$ and $0\le i \le m-k$, we denote by $\operatorname{substr}_k(s,i)$ a length-$k$ substring of~$s$ starting at position $i$,
\[
  \operatorname{substr}_k(s,i) = s_is_{i+1}\cdots s_{i+k-1}.
\]
The set of all length-$k$ substrings of $s$ is called the \emph{$k$-spectrum} of $s$ and is denoted by $\operatorname{Spec}_k(s)$,
\[
    \operatorname{Spec}_k(s) = \{\operatorname{substr}_k(s, i) : 0\le i \le |s|-k \}.
\]

A~string $w$ is~a~\emph{common superstring} of~a~set $S$ if~$s\sqsubseteq w$ for every $s\in S$.
The \problemname{Shortest Common Superstring} problem, abbreviated \problemabbr{SCS}, takes as~input a~finite nonempty set $S$ of~distinct strings over a~finite alphabet $\Sigma$ and asks for a~common superstring of~minimum length.
For an~integer $k\ge 1$, $k$-\problemabbr{SCS} is~the restriction of~\problemabbr{SCS} to~inputs $S\subseteq\Sigma^k$; equivalently, every input string has length exactly $k$.
Consequently, for an~instance $S$ of~$k$-\problemabbr{SCS}, we~have $s\not\sqsubseteq t$ for all distinct $s,t\in S$.
Throughout the paper, $n=|S|$ denotes the number of~input strings.

Starting from $S$, $\operatorname{GREEDY}$ repeatedly chooses an~ordered pair $p,q$ of~distinct strings with maximum overlap and replaces them by~$p\odot q$~\cite{TU88}.
Ties are resolved arbitrarily, and $\operatorname{GREEDY}$ continues to~merge pairs until one string remains.

Let $\operatorname{GREEDY}(S)$ be~the set of~all outputs obtainable in~this way, and let $\operatorname{OPT}(S)$ be~the minimum length of~a~common superstring of~$S$.
We~write
\[\rho_k=\sup_{\substack{\Sigma\ \mathrm{finite}, \\ \varnothing\ne S\subseteq\Sigma^k}}\ \sup_{g\in\operatorname{GREEDY}(S)}\frac{|g|}{\operatorname{OPT}(S)},\]
for the worst-case approximation ratio of the greedy algorithm on strings of length $k$.

For an instance $S$, we write $\|S\|=\sum_{s\in S}|s|$ for its total length. The \emph{compression} of~a~common superstring $w$ of $S$ is~$\|S\|-|w|$.
\begin{theorem}[\cite{TU88,KN24}]\label{thm:half-compression}
	Let $S$ be~a~finite nonempty set of~strings such that $x\not\sqsubseteq y$ for all distinct $x,y\in S$.
	For every $g\in\operatorname{GREEDY}(S)$,
	\[\|S\|-|g|\ge\frac12\bigl(\|S\|-\operatorname{OPT}(S)\bigr).\]
\end{theorem}

For a~string $x$ produced during an~execution of the greedy algorithm from $S$, let $\operatorname{left}(x),\operatorname{right}(x)\in S$ be~its leftmost and rightmost input strings.
\begin{lemma}[\cite{BJLTY91}]\label{lem:endpoint-inheritance}
	For an~instance $S$ of~$k$-\problemabbr{SCS}, at~any point during an~execution of~the greedy algorithm, any two distinct strings $x,y$ currently present satisfy $x\not\sqsubseteq y$ and
	\[\operatorname{ov}(x,y)=\operatorname{ov}(\operatorname{right}(x),\operatorname{left}(y))\le k-1.\]
\end{lemma}

For a~directed graph $H$, write $V(H)$ and $E(H)$ for its vertex and edge sets, respectively, and write $v(H)=|V(H)|$ and $e(H)=|E(H)|$ for their cardinalities.
Its \emph{weak components} are the connected components of~the undirected graph obtained by~dropping the directions of~the edges.
The graph $H$ is~\emph{weakly connected} if~it~has exactly one weak component.
A~\emph{trail} in~a~directed graph is~a~directed walk in~which no~edge is~repeated.
A~\emph{trail cover} of $H$ is~an~edge-disjoint family of~directed trails whose union is~$E(H)$.
Write $\tau(H)$ for the minimum size of~such a~cover.

For an~instance $S\subseteq\Sigma^k$ of~$k$-\problemabbr{SCS}, let $\mathcal G(S)$ be~a~directed graph whose vertices are the length-$(k-1)$ prefixes and suffixes of~the strings in~$S$.
For every string $s = s_0\cdots s_{k-1}\in S$, the graph $\mathcal G(S)$ contains the edge
\[s_0\cdots s_{k-2}\longrightarrow s_1\cdots s_{k-1}.\]
Equivalently, $\mathcal G(S)$ is~the subgraph of~the de Bruijn graph~\cite{Bruijn46} on~the strings of~length $k-1$ over the alphabet~$\Sigma$ consisting of~the edges corresponding to~the strings of~$S$.
\begin{lemma}[\cite{Pevzner89}]\label{lem:eulerian-optimum}
If~$\mathcal{G}(S)$ is~balanced and weakly connected, then
\[
    \operatorname{OPT}(S)=n+k-1.
\]
\end{lemma}


%% file: sections/lower_bound_rho_k.tex
\section{Lower bound for \texorpdfstring{$\rho_k$}{rhok} when \texorpdfstring{$k\ge 6$}{k >= 6}}\label{sec:rhok-lower}

\begin{theorem}\label{thm:rhok-lower}
	For all $k\ge 6$, $\rho_k\ge 2$.
\end{theorem}

We first establish a~construction that shows $\rho_6\ge 2$, and then we extend this construction to~all $k\ge 7$.
The construction is based on the cyclic spectra of strings.

\subsection{Cyclic spectra of~strings}
For a string $s = s_0s_1\dotsc s_{m - 1}$ of~length $m\ge1$ and integers $k\ge 1$ and $0\le i < m$, we denote by $\operatorname{csubstr}_k(s,i)$ the length-$k$ \emph{cyclic substring} of $s$ starting at position $i$: 
\[ 
    \operatorname{csubstr}_k(s,i) = s_i s_{i+1}\cdots s_{i+k-1}, 
\] 
where the indices are taken modulo $m$, that is, $s_j=s_{j\bmod m}$ for every $j\ge 0$.
The set of all length-$k$ cyclic substrings of $s$ is called the \emph{cyclic $k$-spectrum} of $s$ and is denoted by $\operatorname{CSpec}_k(s)$: 
\[ 
    \operatorname{CSpec}_k(s) = \{\operatorname{csubstr}_k(s,i) : 0\le i < |s|\}. 
\]

For some fixed positive integer $k$, let $u$~be a~string of length $n$ and assume that $\operatorname{CSpec}_k(u)$ consists of~$n$ different strings. Let $\ell=n + k - 1$, and for $0 \le i < n$, define a~string $w_i = \operatorname{csubstr}_{\ell}(u,i)$. Clearly, $w_0, \dotsc, w_{n-1}$ are cyclic shifts of~each other.
A~simple but crucial observation is~that $w_i \in \operatorname{GREEDY}(\operatorname{CSpec}_k(u))$ for every $0 \le i < n$.
This~is best illustrated with
an~example. Let $u=\texttt{ABCA}$ and $k=5$. Then, 
\[\operatorname{CSpec}_k(u)=\{\texttt{ABCAA},\texttt{BCAAB},\texttt{CAABC},\texttt{AABCA}\}.\]
Clearly, every two consecutive strings here, including the last one and the first one, have an~overlap of~size~$k-1 = 4$. Hence, starting from any string in~this sequence, the greedy algorithm may merge the strings going (cyclically) from left to~right to~obtain any of~the following superstrings $w_0, w_1, w_2, w_3$ of length $\ell = n+k-1 = 8$:
\[\texttt{ABCAABCA},\ \texttt{BCAABCAA},\ \texttt{CAABCAAB},\ \texttt{AABCAABC}.\]

\subsection{The case of $k=6$}
We~start by~covering the case $k=6$
as~it~demonstrates the main ideas in~a~clean way.
We~then generalize the construction to~every $k \ge 6$.

Consider two strings
$p=\texttt{ABAABX}$ and $q=\texttt{ABX}$
and their cyclic $6$-spectra:
\begin{align*}
	\operatorname{CSpec}_6(p)&=
		\{
			\texttt{ABAABX}, 
			\texttt{BAABXA}, 
			\texttt{AABXAB}, 
			\texttt{ABXABA}, 
			\texttt{BXABAA}, 
			\texttt{XABAAB}
		\},
	\\
	\operatorname{CSpec}_6(q)&=
		\{
			\texttt{ABXABX}, 
			\texttt{BXABXA},
			\texttt{XABXAB}
		\}.
\end{align*}
It~is straightforward to~check that all nine strings are distinct; hence,
$S=\operatorname{CSpec}_6(p) \cup \operatorname{CSpec}_6(q)$ has size~9. Moreover, $S$~is a~(linear) $6$-spectrum of~the following string of~length~$14$:
\[s^*=\texttt{ABAABXABXABAAB}.\]
It~is not difficult to~check this by~hand, but it~is particularly instructive
to~look~at the graph~$\mathcal{G}(S)$. 

\begin{center}
	\begin{tikzpicture}[>=latex, yscale=.4]
		\tikzstyle{v}=[rectangle, draw, inner sep=.8mm]
		
		\node[v] (ABXAB) at (0, 0)       {\texttt{ABXAB}};
		\begin{scope}[shift={(2, 0)}]
			\node[v] (BXABX) at (120:2)  {\texttt{BXABX}}; 
			\node[v] (XABXA) at (-120:2) {\texttt{XABXA}}; 
		\end{scope}
		\begin{scope}[shift={(-2, 0)}]
			\node[v] (BXABA) at (60:2)  {\texttt{BXABA}}; 
			\node[v] (XABAA) at (120:2) {\texttt{XABAA}}; 
			\node[v] (ABAAB) at (180:2) {\texttt{ABAAB}};
			\node[v] (BAABX) at (240:2) {\texttt{BAABX}}; 
			\node[v] (AABXA) at (300:2) {\texttt{AABXA}}; 
		\end{scope}
		
		\foreach \f/\t in {ABXAB/BXABX, BXABX/XABXA, XABXA/ABXAB,
			ABXAB/BXABA, BXABA/XABAA, XABAA/ABAAB, ABAAB/BAABX, BAABX/AABXA, AABXA/ABXAB}
			\draw[->] (\f) -- (\t);
	\end{tikzpicture}
\end{center}
It~is clearly Eulerian and using the node \texttt{ABAAB} as a~starting point, one can spell the superstring~$s^*$.
Now, it~is intuitive that, when starting from \texttt{ABAAB} and moving along the left cycle, one needs to~traverse the right cycle after arriving at~the node
\texttt{ABXAB}, and only then continue traversing the left cycle. The greedy algorithm, however, does not have this intuition. Instead, it~may merge the strings from $\operatorname{CSpec}_6(p)$ to~\texttt{XABAABXABAA} and merge the strings from $\operatorname{CSpec}_6(q)$ to~\texttt{BXABXABX}. The resulting two strings of~length~$11$ and~$8$ have an~overlap of~length~1, so the final superstring has length~$18$:
\[\texttt{BXABXABXABAABXABAA}.\]

Again, the fact that the string~$s^*$ has an~overlap of~length~$5$ with itself allows~us to~iterate this construction as~follows. For an~integer parameter~$t\ge1$ and for every $i \in [t]$, let
$p_i=\texttt{ABAABX}_i$ and $q_i=\texttt{ABX}_i$.
That~is, the strings $p_i$ and $q_i$ contain a~fresh symbol $\texttt{X}_i$. Then,
let
\[S_t = \bigcup_{i \in [t]} \operatorname{CSpec}_6(p_i) \cup \operatorname{CSpec}_6(q_i).\]
As~discussed above, each block has a~superstring of~length~$14$ starting and ending with \texttt{ABAAB}. All these strings can be~merged into a~superstring for~$S_t$ of~length $9t+5$. At~the same time, the greedy algorithm may produce 
a~superstring of~length~$18t$. Hence, $\rho_6 \ge \frac{18t}{9t+5}$ meaning that $\rho_6 \ge 2$.

\subsection{The case of $k\ge7$}
Fix any $k\ge 7$, let $d = k - 6$ and choose a string $y = \texttt{Y}_1\texttt{Y}_2\cdots\texttt{Y}_d$ of $d$ fresh symbols.
Consider the following two strings
\[
    p = \texttt{ABAABX}y, \quad q = \texttt{ABX}y.
\]
It is not difficult to see that the same structure as in the case of $k=6$ is preserved: $|p| = k$, $|q|=k-3$, all the strings in their cyclic $k$-spectra are distinct, and hence $S=\operatorname{CSpec}_k(p) \cup \operatorname{CSpec}_k(q)$ has size~$2k-3$. The graph $\mathcal{G}(S)$ also has two cycles that are connected through the vertex $\texttt{ABX}y\texttt{AB}$ (adding the string~$y$ makes the cycles longer while preserving the global structure of the graph). Therefore, there is an optimal common superstring of~$S$ of~length~$3k-4$,
\[
    s^*=y\texttt{ABAABX}y\texttt{ABX}y\texttt{ABAAB}.
\]

The greedy algorithm may merge the strings from $\operatorname{CSpec}_k(p)$ to~$\texttt{X}y\texttt{ABAABX}y\texttt{ABAA}$ and merge the strings from $\operatorname{CSpec}_k(q)$ to~$\texttt{BX}y\texttt{ABX}y\texttt{ABX}$. The resulting two strings of~length~$2k-1$ and~$2k-4$ have an~overlap of~length~$1$, so the final common superstring has length~$4k-6$.

We use the same idea to~iterate this construction: for an~integer parameter~$t\ge 1$ and for every $i \in [t]$, let 
$p_i=\texttt{ABAABX}_iy$ and $q_i=\texttt{ABX}_iy$ for a~fresh symbol $\texttt{X}_i$. Let
\[S_t = \bigcup_{i \in [t]} \operatorname{CSpec}_k(p_i) \cup \operatorname{CSpec}_k(q_i).\]
Each block $i$ has a~superstring of~length~$3k-4$ starting and ending with $y\texttt{ABAAB}$. All these block superstrings can be~merged into a~superstring for~$S_t$ of~length $(2k-3)t+k-1$. At~the same time, the greedy algorithm may produce
a~superstring of~length~$(4k-6)t$ by combining non-overlapping superstrings of length $4k-6$ for each block. Hence,
\[
    \rho_k \ge \frac{(4k-6)t}{(2k-3)t+k-1}.
\] 
Since $t\ge 1$ was arbitrary, we obtain $\rho_k \ge 2$.

%% file: sections/lower_bound_rho_3.tex
\section{Lower bound for \texorpdfstring{$\rho_3$}{rho3}}\label{sec:rho3-lower}


%
%
%

\begin{theorem}\label{thm:lower_3}
	$\rho_3 \ge 9/5$.
\end{theorem}
\begin{proof}
	The main building block of~our construction is a~string \texttt{ABXAXAB} and its 3-spectrum \texttt{ABX}, \texttt{BXA}, \texttt{XAX}, \texttt{AXA}, \texttt{XAB}. 
	For this instance, the greedy algorithm may proceed as~follows: 
	merge \texttt{AXA} and \texttt{XAX} into \texttt{AXAX}; then, merge \texttt{BXA}, \texttt{XAB}, and \texttt{ABX} into \texttt{BXABX}; finally, concatenate the two remaining strings \texttt{AXAX} and \texttt{BXABX}. Thus, instead of~using four $2$-overlaps, it~uses three $2$-overlaps, thereby producing a~superstring of~length~$9$ 
	instead of~an~optimum superstring of~length~$7$.
	
	The fact that the string \texttt{ABXAXAB} has an~overlap of~length~$2$ with itself
	allows~us to~iterate this construction efficiently. Namely, consider a~string
	(square brackets indicate parameterized blocks)
	\[
		\texttt{AB}
		\left[\texttt{X}_1\texttt{AX}_1\texttt{AB}\right]
		\left[\texttt{X}_2\texttt{AX}_2\texttt{AB}\right]
		\dotsb
		\left[\texttt{X}_t\texttt{AX}_t\texttt{AB}\right]
	\]
	of~length $5t+2$ and its 3-spectrum consisting of~$5t$ different 3-strings. 
	The greedy algorithm may take $3t$~overlaps of~length~$2$ to~get the following $2t$~strings:
	\[
		\texttt{AX}_1\texttt{AX}_1,\, \texttt{BX}_1\texttt{ABX}_1,\,
		\texttt{AX}_2\texttt{AX}_2,\, \texttt{BX}_2\texttt{ABX}_2,\, 
		\dotsc,
		\texttt{AX}_t\texttt{AX}_t,\, \texttt{BX}_t\texttt{ABX}_t.
	\]
	For these strings, no~nontrivial overlap remains, so~the greedy algorithm concatenates them and gets a~superstring of~length~$9t$. Thus, $\rho_3 \ge \frac{9t}{5t+2}$, for all~$t$, and hence $\rho_3 \ge \frac{9}{5}$.
\end{proof}

%% file: sections/upper_bound_rho_3.tex
\section{Upper bound for \texorpdfstring{$\rho_3$}{rho3}}\label{sec:rho3-upper}

\begin{theorem}\label{thm:rho3-upper}
	$\rho_3\le\frac95$.
\end{theorem}

\subsection{Graph representation}

Fix an~instance $S$ of~$3$-\problemabbr{SCS} and an~output $g\in\operatorname{GREEDY}(S)$.
By~\Cref{lem:endpoint-inheritance}, overlaps never exceed two and depend only on~the leftmost and rightmost input strings.
So, the greedy algorithm always performs all overlap-two merges before any overlap-one merge, and all overlap-one merges before any zero-overlap merge.
Let $\kappa_2$ be~the number of~strings remaining after all overlap-two merges, and let $\kappa_1$ be~the number of~strings remaining after all subsequent overlap-one merges; thus, $\kappa_2\ge\kappa_1$.
There are $n-\kappa_2$ overlap-two merges and $\kappa_2-\kappa_1$ overlap-one merges.
Starting from total length $3n$, we~therefore obtain
\[|g|=3n-2(n-\kappa_2)-(\kappa_2-\kappa_1)=n+\kappa_1+\kappa_2.\]

Let $\mathcal G=\mathcal G(S)$.
Every overlap-two merge concatenates two trails in~$\mathcal G$; hence, the greedy algorithm forms a~trail cover $C$ of $\mathcal{G}$ in~which no~two distinct trails can be~concatenated, and $|C|=\kappa_2$.
For~$0 \le i \le e(P)$, we denote by $P_i \in \Sigma^2$ the~$i$-th vertex of~the trail $P$ in~$\mathcal{G}$.
The trail $P$ is~\emph{open} if~$P_0\ne P_{e(P)}$ and \emph{closed} otherwise.
For closed $P$, the vertex $P_0=P_{e(P)}$ is~the \emph{root} of~$P$.
We~call all trail vertices \emph{internal} except for the first and last vertices of~the trail.

Let $\mathcal M$ be~a~directed multigraph with vertex set $\Sigma$.
Consecutive vertices of~any $P \in C$ overlap in~one character, so~there are letters $a_0,\ldots,a_{e(P)+1} \in \Sigma$ such that
\[P_i=a_i a_{i+1}\quad(0\le i\le e(P)).\]
For each $P\in C$, the graph $\mathcal M$ has an edge
\[e_P\colon\ a_0\to a_{e(P)+1}.\]
The~overlap-one phase yields a~trail cover of~$\mathcal{M}$ with~$\kappa_1$ trails.

We~start with a~simple lower bound on~the size of~an~optimal string in~terms of~$\tau(\mathcal G)$.
\begin{lemma}\label{lem:opt-lower}
	\[\operatorname{OPT}(S)\ge n+\tau(\mathcal G)+1.\]
\end{lemma}
\begin{proof}
	Let $w$ be~a~shortest common superstring of~$S$, and choose one occurrence of~each input string in~$w$.
	Apply the same two-stage grouping to~these occurrences in~their order, and let $\gamma_2$ and~$\gamma_1$ denote the numbers of~groups remaining after the overlap-two and overlap-one stages, respectively.
	So, $\gamma_2\ge\tau(\mathcal G)$, while $\gamma_1\ge1$ and therefore
	\[\operatorname{OPT}(S)=|w|\ge n+\gamma_2+\gamma_1\ge n+\tau(\mathcal G)+1.\qedhere\]
\end{proof}

\subsection{Proof overview}
The~greedy algorithm has two main phases: first, it merges strings with overlap~two, and then it merges the resulting strings with overlap~one.
The~proof analyzes these phases separately.
The~graph $\mathcal G$ captures the first phase, whereas the graph~$\mathcal M$ captures the second.
Thus, $\kappa_2$ is~exactly the size of~the trail cover $C$ of~$\mathcal G$, while $\kappa_1$ is~the size of~the trail cover produced in~$\mathcal M$.
In~the lower-bound construction from~\Cref{thm:lower_3}, the greedy output contains many length-two closed trails corresponding to~$\texttt{AX}_i\texttt{AX}_i$.
The~following proposition shows that, without closed trails of~length one or~two, the approximation ratio of~the greedy algorithm is~already below $\frac53$.

\begin{proposition}
	Suppose that every closed trail in~$C$ has at~least three edges.
	Then
	\[\frac{|g|}{\operatorname{OPT}(S)}<\frac53.\]
\end{proposition}
\begin{proof}
	Let $\beta$ and~$\gamma$ be~the numbers of~open and closed trails in~$C$, respectively.
	Every open trail contains at~least one edge, while every closed trail contains at~least three, and hence
	\[n\ge\beta+3\gamma=3\kappa_2-2\beta.\]
	Since $\kappa_1\le\kappa_2$,
	\[|g|=n+\kappa_1+\kappa_2\le n+2\kappa_2\le\frac{5n+4\beta}{3}.\]
	By~\Cref{lem:pair-component-accounting}, $\beta\le\tau(\mathcal G)$, and~\Cref{lem:opt-lower} gives $\operatorname{OPT}(S)\ge n+\beta+1$.
	Therefore
	\[\frac{|g|}{\operatorname{OPT}(S)}\le\frac{5n+4\beta}{3(n+\beta+1)}=\frac53-\frac{\beta+5}{3(n+\beta+1)}<\frac53.\qedhere\]
\end{proof}

The general case is~harder precisely because $C$ may contain many short closed trails.
The proof treats such trails as objects that must be paid for by longer trails.
The~payment mechanism is~a~budget.
A~nonisolated loop $[aaa]$ is~assigned to~an~internal occurrence of~the vertex $aa$ on~another trail so~that it~is~accounted for in~that trail's budget.
For every trail $P$ that is~not a~loop, its \emph{budget} $\delta_P$ is~the number of~internal positions of~$P$ that remain available after the loop assignments, with one unit subtracted when $P$ is~closed.
The~total available budget is~$\sigma=\sum_P\delta_P$, where the sum ranges over all non-loop trails.

According to~their role in this accounting, the trails of~$\mathcal{G}$ fall into three types.
\begin{itemize}
	\item A~loop has no~budget of~its own.
	\item A~\emph{bundle} consists of~closed trails that share a~rigid structure and have zero budget: after the loop assignments, each has exactly one available internal position.
	\item A~\emph{connector} is~an~open trail or~a~closed trail with positive budget; connectors can pay for bundles and loops.
\end{itemize}

Consider an~$m$-edge trail.
Since all its intermediate strings have overlap two on~both sides, the trail arranges at~least $m-2$ strings optimally.
The budget of~a~connector measures how much of~this slack remains available to~pay for loops and bundles that the greedy algorithm left separate.

The budget links the two phases: if~the overlap-two phase leaves many trails, then little budget remains, and this limits the cost of~the overlap-one phase.
More precisely, we~convert the trail cover of~$\mathcal G$ into a~trail cover of~$\mathcal M$, with an~overhead controlled by~$\sigma$.
This accounting yields the two complementary inequalities
\[
    2\kappa_2\le n+\tau(\mathcal G)-\sigma,\qquad 
    \tau(\mathcal M)\le2\sigma+2\tau(\mathcal G).
\]
The~first inequality says that a~large budget can exist only when the overlap-two phase leaves few trails, so~a~large $\sigma$ means that this phase is~efficient.
The~second inequality uses the assignments above to~build trail covers of~$\mathcal M$, so~a~small $\sigma$ shows that the overlap-one phase is~efficient.

\subsection{Trail bounds}

The next lemma bounds the number of~strings remaining after the overlap-one phase.
Starting from $\kappa_2$ strings, the greedy algorithm produces a~trail cover of~$\mathcal M$ with at~most $\frac{\kappa_2+\tau(\mathcal M)}{2}$ trails.

\begin{lemma}\label{lem:terminal-packing}
	\[2\kappa_1\le\kappa_2+\tau(\mathcal M).\]
\end{lemma}
\begin{proof}
	Let $R$ be~the set of~$\kappa_2$ strings after the overlap-two phase.
	By~\Cref{lem:endpoint-inheritance}, no~string in~$R$ contains another, and every overlap between distinct strings in~$R$ has length at~most one.
	The overlap-one phase on~$R$ has compression $\kappa_2-\kappa_1$.
	A~minimum trail cover of~$\mathcal M$ orders these $\kappa_2$ strings along $\tau(\mathcal M)$ trails, yielding a~common superstring of~$R$ with compression $\kappa_2-\tau(\mathcal M)$.
	Hence, by~\Cref{thm:half-compression},
	\[\kappa_2-\kappa_1\ge\frac12\bigl(\|R\|-\operatorname{OPT}(R)\bigr)\ge\frac12\bigl(\kappa_2-\tau(\mathcal M)\bigr).\]
	Rearranging gives $2\kappa_1\le\kappa_2+\tau(\mathcal M)$.
\end{proof}

By~the lemmas above, it~remains to~bound $\kappa_2$ and $\tau(\mathcal M)$ in~terms of~$n$ and $\tau(\mathcal G)$.
We~use a~nonnegative budget $\sigma$ and show that a~large $\sigma$ forces $\kappa_2$ to~be~small, whereas a~small $\sigma$ yields a~small cover of~$\mathcal M$.

We~begin by~accounting for the contribution of~each weak component to~$\tau(\mathcal G)$.
Let $H_1,\ldots,H_\nu$ be~the weak components of~$\mathcal G$, and for $1\le j\le\nu$ let
\[\beta_j=\sum_{\substack{v\in V(H_j), \\ \deg^+(v) \ge \deg^-(v)}}\deg^+(v)-\deg^-(v).\]
At~a~vertex $v$, the~imbalance $\deg^+(v)-\deg^-(v)$ equals the~number of~trails that start at~$v$ minus the~number of~trails that end at~$v$, for any trail cover of~$\mathcal{G}$.

We use the following result due to~Gallant, Maier, and Storer~\cite{GMS80}.
It~shows that the degree imbalances fully determine the number of~open trails in~$C$, whereas the number of~closed trails may depend on~the greedy choices.
Thus, the main difficulty is~to~control the closed trails.
\begin{lemma}[\cite{GMS80}]\label{lem:pair-component-accounting}
	$\tau(\mathcal G)=\sum_{j=1}^{\nu}\max\{\beta_j,1\}$.
	Moreover, each component $H_j$ of~$\mathcal G$ contains exactly $\beta_j$ open trails of~$C$.
\end{lemma}
%

We now define the budget formally. Set
\[C^*=\{P\in C:P\text{ is~open or~}e(P)>1\}.\]
The trails omitted from $C^*$ are precisely the loops $[aaa]$.
For a~closed $P \in C^*$, an~index $i \in [0, e(P) - 1]$ is~a~\emph{transition} if~$a_i\ne a_{i+1}$, where $P_i = a_i a_{i + 1}$.
Every other index $i$ is~called a~\emph{constant position}.
A~string is~called \emph{constant} if~it~has no~transitions.
A~transition at~$i=0$ is~said to~occur at~the root.
For every loop $[aaa]\in C$ whose weak component in~$\mathcal G$ contains another edge, there is~some other trail $P\in C$ that visits $aa$.
Note that $e(P) > 1$, and every occurrence $P_i=aa$ is~internal, since otherwise the greedy algorithm would concatenate $P$ and $[aaa]$.
Choose one such occurrence and say that $[aaa]$ is~assigned to~$P_i$.

For $P\in C^*$, let $\operatorname{free}(P)$ be~the set of~internal positions not used by~the loop assignment:
\[\operatorname{free}(P)=\{i:0<i<e(P)\text{ and no~loop is~assigned to~}P_i\}.\]
The number of~free positions will bound the number of~additional trails needed in~our construction of~a~trail cover of~$\mathcal M$.
Define
\[
    \delta_P=
    \begin{cases}
        |\operatorname{free}(P)|,&P\text{ is open},\\|\operatorname{free}(P)|-1,&P\text{ is closed}.
    \end{cases}
\]
For a~closed $P\in C^*$, the sequence $a_0,\ldots,a_{e(P)}$ contains at~least two transitions because $P$ is~a~trail with more than one edge, and $a_{e(P)}=a_0$ because $P$ is~closed.
At~most one of~them occurs at~the root, and every other transition belongs to~$\operatorname{free}(P)$; hence $\operatorname{free}(P)\ne\varnothing$ and $\delta_P\ge0$.
We~refer to~$\delta_P$ as~the budget of~$P$.
Finally, set
\[\sigma=\sum_{P\in C^*}\delta_P,\]
which is~the total budget over all trails in~$C^*$.

\begin{lemma}\label{lem:trail-accounting}
	\[2\kappa_2+\sigma\le n+\tau(\mathcal G).\]
\end{lemma}
\begin{proof}
	Let $\beta=\sum_{j=1}^{\nu}\beta_j$ be~the number of~open trails in~$C$, and let $\lambda_0$ be~the number of~loops that form isolated weak components of~$\mathcal G$.
	By~the definition of~$\delta_P$, the contribution of~an~open trail $P$ to~$2\kappa_2+\sigma$ is~$e(P)+1$ minus the number of~loops assigned to~$P$.
	For a~closed trail $P\in C^*$, the contribution is~$e(P)$ minus the number of~loops assigned to~$P$.
	A~loop contributes $2$ to~$2\kappa_2+\sigma$ but only $1$ to~$n$.
	If~the loop is~nonisolated, then it~is~assigned to~a~trail $P$, and this assignment decreases $\delta_P$ by~$1$.
	Its contribution to~$2\kappa_2+\sigma$ is~therefore $1$, matching its contribution to~$n$.
	An~isolated loop has no~assignment, so~the difference remains and is~counted by~$\lambda_0$.
	By~\Cref{lem:pair-component-accounting}, $\beta+\lambda_0\le\tau(\mathcal G)$, and hence
	\[2\kappa_2+\sigma=n+\beta+\lambda_0\le n+\tau(\mathcal G).\qedhere\]
\end{proof}

Now, we construct a~trail cover of~$\mathcal M$.
We~first show that all closed trails with zero budget have a~rigid structure.

\begin{lemma}\label{lem:zero-delta-shape}
	If~$P\in C^*$ is~closed and $\delta_P=0$, then $P$ has exactly two transitions, one of~which occurs at~the root.
	Consequently, exactly two distinct letters occur among $a_0,\ldots,a_{e(P)-1}$, and every position $i$ with $a_i=a_{i+1}$ receives an~assigned loop.
\end{lemma}
\begin{proof}
	Let
	\[\eta_P=\bigl|\{0\le i<e(P):a_i\ne a_{i+1}\}\bigr|.\]
	Since $a_{e(P)}=a_0$ and the sequence $a_0,\ldots,a_{e(P)}$ is~nonconstant, we~have $\eta_P\ge2$.
	Since $\delta_P=0$, exactly $e(P)-2$ loops are assigned to~$P$.
	These loops occupy distinct constant positions, so~$e(P)-2\le e(P)-\eta_P$ and hence $\eta_P\le2$.
	Thus $\eta_P=2$, and every constant position receives an~assigned loop.
	The root cannot be~assigned a~loop, so~it~is~a~transition.
\end{proof}

Although zero-budget trails cannot pay for other pieces, they have a~particularly simple joint structure in~$\mathcal M$.
We~group them into the following objects.
For an~unordered pair $\{a,b\} \subseteq \Sigma$ with $a\ne b$, set
\[Z_{a,b}=\left\{P\in C^*:P\text{ is closed},\ \delta_P=0,\ \{a_i:0\le i<e(P)\}=\{a,b\}\right\},\]
and call every nonempty $Z_{a,b}$ a~\emph{bundle}.
For a~trail $P$, let $V(P)$ be~the set of~vertices of~$\mathcal G$ visited by~$P$, and put $V(Z)=\bigcup_{P\in Z}V(P)$ for a~bundle $Z$.

By~\Cref{lem:zero-delta-shape}, every trail in~$Z_{a,b}$ is~rooted at~$ab$ or~$ba$.
At~most one trail has either root, because two closed trails with the same root could be~concatenated; hence, $|Z_{a,b}|\le2$.
Also, distinct bundles have disjoint vertex sets.

More importantly, the edges of~$\mathcal M$ corresponding to~a~bundle $Z_{a,b}$ belong to~$\{a\to b,b\to a\}$, with at~most one edge in~each direction, and its assigned loops give only the edges $a\to a$ and $b\to b$.
Thus an~entire bundle, together with its assigned loops, can be~covered by~a~single trail in~$\mathcal M$.

\subsection{Local covers}

Since a~bundle has no~budget of~its own, we~either leave it~as~the sole nontrivial piece of~its component or~assign it~to~a~connector and include it~in~the corresponding local cover.

Call a~trail $P\in C^*$ a~\emph{connector}\footnote{Each trail $P\in C$ is~either a~loop, an~element of a~bundle, or~a~connector.} if~$P$ is~open or~$\delta_P>0$.
We~now incorporate each bundle that lies in~a~component containing a~connector into~a~trail cover of~$\mathcal M$.
Loops do~not help to connect distinct trails, and distinct bundles do~not intersect.
Thus~every bundle either forms an~entire component together with its assigned loops or~shares a~vertex with a~connector; the only component containing neither is~an~isolated loop.
A~component with no~connector therefore can be~covered by~one trail in~$\mathcal M$.

We~henceforth focus on~components containing connectors.
A~bundle that shares an~endpoint with a~connector will be~inserted there (note that no~two bundles can contain the same endpoint). 
Each remaining bundle $Z$ will be~assigned to~a~free internal position $P_i$ of~a~connector $P$ satisfying $P_i\in V(Z)$.

First, suppose that a~bundle $Z$ and a~connector $P$ share the endpoint $ab$.
Then $a\ne b$, since otherwise $[aaa]$ could concatenate with $P$.
For the same reason, $Z$ does not contain a~trail rooted at~$ab$.
Thus $Z=Z_{a,b}$, and since it~is~nonempty, it~consists of~one trail $Q$ rooted at~$ba$.
In~$\mathcal M$, the edge $e_Q=b\to a$, together with the loops assigned to~$Q$, can be~prepended to~$e_P$ if~$ab$ is~the first vertex of~an~open $P$, and appended otherwise.
In~either case, the result is~a~trail in~$\mathcal M$.

Now consider a~bundle $Z$ that shares no~endpoint with any connector.
Choose a~connector $P$ and an~occurrence $P_i\in V(Z)$.
Then $0<i<e(P)$.
Moreover, $i\in\operatorname{free}(P)$: if~$P_i=ab$ with $a\ne b$, no~loop can be~assigned there, whereas if~$P_i=aa$, the unique loop $[aaa]$ is~assigned to~a~trail in~$Z$ and not to~$P$.
The chosen positions are distinct because distinct bundles have disjoint vertex sets.
For each connector $P$, let $B_P\subseteq\operatorname{free}(P)$ be~the set of~positions assigned to~bundles.
Thus every bundle has either been inserted at~an~endpoint or~assigned to~one position in~$B_P$.

Bundles are not the only pieces that can cost an~additional trail in~$\mathcal M$.
A~loop may be~nonisolated in~$\mathcal G$ but become the only edge of~its weak component in~$\mathcal M$.
Such a~loop cannot be~inserted into another trail of~$\mathcal M$, so~we~also include it~in~the local cover of~the connector to~which it~was assigned.
For each connector $P$, let $\lambda_P$ be~the number of~isolated loops assigned to~$P$.
We~first show that the positions adjacent to~an~isolated loop cannot be~assigned bundles.

\begin{lemma}\label{lem:isolated-loop-charging}
	Fix a~connector $P$ and an~isolated loop $[aaa]$, counted by $\lambda_P$, and let $P_i=aa$ be~its assigned occurrence.
	Since $0<i<e(P)$, let $b=a_{i-1}$ and~$c=a_{i+2}$ be~the letters immediately before and after this occurrence.
	Then $b,c\ne a$, and every bundle $Z$ satisfies
	\[\{ba,ac\}\cap V(Z)=\varnothing.\]
\end{lemma}
\begin{proof}
	If~$b=a$ or~$c=a$, then $P$ contains the edge $[aaa]$, contradicting that this edge belongs to~the distinct trail $[aaa]\in C$.
	Hence $b,c\ne a$.
	Since $e_{[aaa]}  \colon a\to a$ is~an~isolated component of~$\mathcal M$, no~bundle edge is~incident to~$a$.
	If~a~bundle $Z$ visits $ba$ or~$ac$, then the unordered pair defining $Z$ would contain $a$, so~some bundle edge $e_Q$ would be~incident to~$a$.
	This is~a~contradiction, and hence $\{ba,ac\}\cap V(Z)=\varnothing$.
\end{proof}

We~can now show that the bundles and isolated loops fit into the available budget of~their connector.
\begin{lemma}\label{lem:connector-budget}
	Every connector $P$ satisfies
	\[|B_P|+\lambda_P\le\delta_P+1.\]
	Moreover, if~$P$ is~closed, then $\lambda_P\le\delta_P$.
\end{lemma}
\begin{proof}
	For each isolated loop $[aaa]$ counted by~$\lambda_P$, let $j(a)$ denote~the index of~the pair vertex $ac$, where~$c$ is~the letter following~the occurrence $aa$ in~$P$.
	By~\Cref{lem:isolated-loop-charging}, $c\ne a$ and no~bundle visits~$ac$.
	Hence no~loop is~assigned to~the position $j(a)$, and $j(a)\notin B_P$.
	Since $e_{[aaa]}$ is~an~isolated component of~$\mathcal M$, the edge $e_P$ is~not incident with $a$; hence $j(a)$ is~not the root when $P$ is~closed.
	Moreover,
	\[j(a)\in\begin{cases}\operatorname{free}(P)\cup\{e(P)\},&P\text{ is open},\\\operatorname{free}(P),&P\text{ is closed},\end{cases}\]
	and $|\{a:j(a)=e(P)\}|\le1$.
	Counting the indices in~$\operatorname{free}(P)\setminus B_P$ yields
	\[
        |B_P|+\max\{\lambda_P-1,0\}\le\delta_P\quad(P\text{ is open}),\qquad |B_P|+\lambda_P\le\delta_P+1\quad(P\text{ is closed}).
    \]
	Both cases imply $|B_P|+\lambda_P\le\delta_P+1$.
	For closed $P$, the first transition encountered after the root is~a~free position distinct from every $j(a)$, because its preceding letter occurs in~the root and hence is~not an~isolated loop; therefore $|\operatorname{free}(P)|\ge\lambda_P+1$, or~equivalently, $\lambda_P\le\delta_P$.
\end{proof}

We~are ready to construct the local pieces of~a~trail cover of~$\mathcal M$.
For a~connector $P$, we build a~subgraph of $\mathcal{M}$ around it as follows. Start with the edge $e_P$, add all edges from the bundles inserted at~the endpoints of~$P$ or~assigned to~a~position in~$B_P$, include the loops assigned to those bundles, and finally add the isolated loops assigned to~$P$. Denote the resulting subgraph of~$\mathcal M$ by~$\mathcal D_P$.

\begin{lemma}\label{lem:local-connector-cover}
	The graph $\mathcal D_P$ has a~trail cover of~size at~most
	\[\begin{cases}2\delta_P+2,&P\text{ is open},\\2\delta_P,&P\text{ is closed}.\end{cases}\]
\end{lemma}
\begin{proof}
	After the endpoint insertions, one trail contains $e_P$ and all the bundles inserted at~the endpoints of~$P$.
	Each position in~$B_P$ corresponds to~one bundle, which forms one trail together with its assigned loops.
	Finally, each of~the $\lambda_P$ isolated loops contributes one trail.
	Hence
	\[\tau(\mathcal D_P)\le1+|B_P|+\lambda_P\le\delta_P+2.\]
	This proves the claim for open $P$, and for closed $P$ with $\delta_P\ge2$.
	It~remains to~consider a~closed $P$ with $\delta_P=1$.
	Then $|\operatorname{free}(P)|=2$ and $\lambda_P\le1$.

	Suppose first that $\lambda_P=1$, and let $[aaa]$ be~the isolated loop assigned to~$P$.
	Since the assigned occurrence is~$P_i=aa$ with $0<i<e(P)$, let $b=a_{i-1}$ and~$c=a_{i+2}$ be~the letters immediately before and after it.
	Since this loop is~isolated from~$e_P$, the root contains no~$a$.
	Hence the two boundary transitions $ba$ and~$ac$ are distinct positions of~$\operatorname{free}(P)$, and by~\Cref{lem:isolated-loop-charging}, no~bundle visits either of~them.
	They exhaust $\operatorname{free}(P)$, so~$B_P=\varnothing$ and $\tau(\mathcal D_P)\le2$.

	Now suppose that $\lambda_P=0$.
	If~$|B_P|\le1$, the trivial cover has size at~most two, so~assume that $|B_P|=2$.
	Every transition of~$P$ other than the root belongs to~$\operatorname{free}(P)$, so~$P$ has at~most three transitions.
	Since $P$ has at~least two transitions, it~has either two or~three.
	Suppose that it~has exactly two, say $ab$ and~$ba$.
	If~the root is~the transition $ab$, any bundle containing $ba$ also contains the root and was already inserted there, so~the position of~$ba$ does not belong to~$B_P$.
	If~the root is~constant, the vertices $ab$ and $ba$ both belong to~$V(Z_{a,b})$, so~at~most one of~their positions belongs to~$B_P$.
	Thus $|B_P|\le1$, contradicting $|B_P|=2$.
	Hence $P$ has exactly three transitions: the root and the two positions in~$B_P$.

	They use three distinct letters and form the three sides of~a~triangle in~$\mathcal D_P$.
	Hence $\mathcal D_P$ is~weakly connected with total positive imbalance at~most two, and the construction from~\Cref{lem:pair-component-accounting} gives~$\tau(\mathcal D_P)\le2$.
\end{proof}

\subsection{Combining the local covers}

The previous subsection covers the edges that are attached to individual connectors.
We~now sum those local covers component by~component in~$\mathcal{G}$ and insert any remaining nonisolated loops for free.

\begin{lemma}\label{lem:M-cover}
	The graph $\mathcal M$ satisfies
	\[\tau(\mathcal M)\le 2\sigma+2\tau(\mathcal G).\]
\end{lemma}
\begin{proof}
	For $1\le j\le\nu$, set
	\[\sigma_j=\sum_{\substack{P\in C^*,\\E(P)\subseteq E(H_j)}}\delta_P.\]
	By~\Cref{lem:local-connector-cover}, each connector $P$ contributes at~most $2\delta_P$ trails, and an~open connector contributes two additional trails.
	By~\Cref{lem:pair-component-accounting}, the component $H_j$ contains exactly $\beta_j$ open connectors.
	Hence, if~$H_j$ contains a~connector, the subgraphs $\mathcal D_P$ associated with its connectors together have a~trail cover of~size at~most $2\sigma_j+2\beta_j$.

	If~$H_j$ contains no~connector, then~it~consists of~one bundle together with its assigned loops, or~of~one isolated loop; in~either case, the corresponding edges of~$\mathcal M$ have a~one-trail cover.
	Thus, in~all cases, the local pieces arising from~$H_j$ can be~covered by~at~most $2\sigma_j+2\max\{\beta_j,1\}$ trails.

	The only uncovered edges are loops $a\to a$ whose weak components in~$\mathcal M$ contain another edge.
	For each such loop, some already covered edge is~incident to~$a$, so~one of~the trails in~the current cover visits $a$.
	We~can insert the loop at~that visit without increasing the number of~trails.

	Summing over all the components and applying~\Cref{lem:pair-component-accounting} gives
	\[\tau(\mathcal M)\le2\sum_{j=1}^{\nu}\sigma_j+2\sum_{j=1}^{\nu}\max\{\beta_j,1\}=2\sigma+2\tau(\mathcal G).\qedhere\]
\end{proof}

\begin{proof}[Proof of~\Cref{thm:rho3-upper}]
	Substituting the upper bound on~$2\sigma$ from~\Cref{lem:trail-accounting} into~\Cref{lem:M-cover} eliminates $\sigma$ and gives
	\[4\kappa_2+\tau(\mathcal M)\le2n+4\tau(\mathcal G).\]
	Rearranging and applying~\Cref{lem:terminal-packing} gives
	\[\kappa_1\le n+2\tau(\mathcal G)-\frac32\kappa_2.\]
	To~cancel the $\kappa_2$ term, we~combine four fifths of~this bound with one fifth of~the trivial bound $\kappa_1\le\kappa_2$:
	\[\kappa_1+\kappa_2=\frac45\kappa_1+\frac15\kappa_1+\kappa_2\le\frac45\left(n+2\tau(\mathcal G)-\frac32\kappa_2\right)+\frac15\kappa_2+\kappa_2=\frac45n+\frac85\tau(\mathcal G).\]
	Together with~\Cref{lem:opt-lower}, this yields
	\[
        \frac{|g|}{\operatorname{OPT}(S)}\le\frac{9n+8\tau(\mathcal G)}{5(n+\tau(\mathcal G)+1)}=\frac95-\frac{\tau(\mathcal G)+9}{5(n+\tau(\mathcal G)+1)}<\frac95.\qedhere
    \]
\end{proof}

%% file: refs.bib
@article{Bruijn46,
  author  = {Nicolaas Govert de Bruijn},
  title   = {{A Combinatorial Problem}},
  journal = {Proceedings of the Section of Sciences of the Koninklijke Nederlandse Akademie van Wetenschappen te Amsterdam},
  volume  = {49},
  number  = {7},
  pages   = {758--764},
  year    = {1946}
}

@article{Graham66,
  author  = {Ronald L. Graham},
  title   = {{Bounds for Certain Multiprocessing Anomalies}},
  journal = {Bell System Technical Journal},
  volume  = {45},
  number  = {9},
  pages   = {1563--1581},
  year    = {1966},
  doi     = {10.1002/j.1538-7305.1966.tb01709.x}
}

@inproceedings{Johnson73,
  author       = {David S. Johnson},
  title        = {Approximation Algorithms for Combinatorial Problems},
  booktitle    = {{STOC}},
  pages        = {38--49},
  publisher    = {{ACM}},
  year         = {1973}
}

@article{GMS80,
  author       = {John Gallant and
                  David Maier and
                  James A. Storer},
  title        = {On Finding Minimal Length Superstrings},
  journal      = {J. Comput. Syst. Sci.},
  volume       = {20},
  number       = {1},
  pages        = {50--58},
  year         = {1980}
}

@book{Storer88,
  author    = {James A. Storer},
  title     = {{Data Compression: Methods and Theory}},
  publisher = {Computer Science Press},
  year      = {1988}
}

@article{TU88,
  author  = {Jorma Tarhio and Esko Ukkonen},
  title   = {{A Greedy Approximation Algorithm for Constructing Shortest Common Superstrings}},
  journal = {Theor. Comput. Sci.},
  volume  = {57},
  number  = {1},
  pages   = {131--145},
  year    = {1988},
  doi     = {10.1016/0304-3975(88)90167-3}
}

@article{Turner89,
  author  = {Jonathan S. Turner},
  title   = {{Approximation Algorithms for the Shortest Common Superstring Problem}},
  journal = {Inf. Comput.},
  volume  = {83},
  number  = {1},
  pages   = {1--20},
  year    = {1989},
  doi     = {10.1016/0890-5401(89)90044-8}
}

@article{Pevzner89,
  title={l-tuple DNA sequencing: computer analysis},
  author={Pevzner, Pavel A},
  journal={Journal of Biomolecular structure and dynamics},
  volume={7},
  number={1},
  pages={63--73},
  year={1989},
  publisher={Taylor \& Francis}
}

@article{Ukkonen90,
  author  = {Esko Ukkonen},
  title   = {{A Linear-Time Algorithm for Finding Approximate Shortest Common Superstrings}},
  journal = {Algorithmica},
  volume  = {5},
  number  = {3},
  pages   = {313--323},
  year    = {1990},
  doi     = {10.1007/BF01840391}
}

@inproceedings{BJLTY91,
  author       = {Avrim Blum and
                  Tao Jiang and
                  Ming Li and
                  John Tromp and
                  Mihalis Yannakakis},
  title        = {Linear Approximation of Shortest Superstrings},
  booktitle    = {{STOC}},
  pages        = {328--336},
  publisher    = {{ACM}},
  year         = {1991}
}

@book{Vazirani01,
  author    = {Vijay V. Vazirani},
  title     = {{Approximation Algorithms}},
  publisher = {Springer},
  year      = {2001}
}

@inproceedings{RBT04,
  author       = {Heidi J. Romero and
                  Carlos A. Brizuela and
                  Andrei Tchernykh},
  title        = {An Experimental Comparison of Approximation Algorithms for the Shortest
                  Common Superstring Problem},
  booktitle    = {{ENC}},
  pages        = {27--34},
  publisher    = {{IEEE} Computer Society},
  year         = {2004}
}

@article{Ma09,
  author  = {Bin Ma},
  title   = {{Why Greed Works for Shortest Common Superstring Problem}},
  journal = {Theor. Comput. Sci.},
  volume  = {410},
  number  = {51},
  pages   = {5374--5381},
  year    = {2009},
  doi     = {10.1016/j.tcs.2009.09.014}
}

@incollection{GP14,
  author    = {Theodoros P. Gevezes and Leonidas S. Pitsoulis},
  title     = {{The Shortest Superstring Problem}},
  booktitle = {Optimization in Science and Engineering},
  pages     = {189--227},
  publisher = {Springer},
  year      = {2014},
  doi       = {10.1007/978-1-4939-0808-0_10}
}

@article{CR18,
  author       = {Bastien Cazaux and
                  Eric Rivals},
  title        = {Relationship between superstring and compression measures: New insights
                  on the greedy conjecture},
  journal      = {Discret. Appl. Math.},
  volume       = {245},
  pages        = {59--64},
  year         = {2018}
}

@inproceedings{EMV23,
  author       = {Matthias Englert and
                  Nicolaos Matsakis and
                  Pavel Vesel{\'{y}}},
  title        = {Approximation Guarantees for Shortest Superstrings: Simpler and Better},
  booktitle    = {{ISAAC}},
  series       = {LIPIcs},
  volume       = {283},
  pages        = {29:1--29:17},
  publisher    = {Schloss Dagstuhl - Leibniz-Zentrum f{\"{u}}r Informatik},
  year         = {2023}
}

@inproceedings{Nikolaev24,
  author       = {Maksim S. Nikolaev},
  title        = {Greedy Conjecture for the Shortest Common Superstring Problem and
                  Its Strengthenings},
  booktitle    = {{SPIRE}},
  series       = {Lecture Notes in Computer Science},
  volume       = {14899},
  pages        = {233--248},
  publisher    = {Springer},
  year         = {2024}
}

@inproceedings{KN24,
  author       = {Pavel E. Kalugin and
                  Maksim S. Nikolaev},
  title        = {The greedy algorithm for the Shortest Common Superstring problem is
                  a {\textonehalf}-approximation in terms of compression: a simple proof},
  booktitle    = {{SOSA}},
  pages        = {97--99},
  publisher    = {{SIAM}},
  year         = {2024}
}
